\documentclass[12pt]{amsart}
\usepackage{amssymb,amscd}
\usepackage{verbatim}

\usepackage{xcolor}

\usepackage{amsmath,amssymb,graphicx,mathrsfs}   
\usepackage{enumerate}
\usepackage[colorlinks=true,allcolors = blue]{hyperref} 

\usepackage{tikz}
\usetikzlibrary{matrix}

\usepackage[all]{xy}

\usepackage{amssymb,amsfonts,amsthm,amsmath,calligra}
\usepackage{slashed}
\usepackage{yfonts}
\usepackage{mathrsfs,pifont}
\usepackage{float}

\usepackage[all]{xy}

\usepackage{tikz}

\usepackage{graphicx}
\usepackage{xcolor}

\usepackage{amssymb,amsfonts,amsthm,amsmath}
\usepackage[all]{xy}
\usepackage{slashed}
\usepackage{yfonts}
\usepackage{mathrsfs,pifont}

\let\frak\mathfrak

\def\>{\relax\ifmmode\mskip.666667\thinmuskip\relax\else\kern.111111em\fi}
\def\<{\relax\ifmmode\mskip-.333333\thinmuskip\relax\else\kern-.0555556em\fi}
\def\vsk#1>{\vskip#1\baselineskip}
\def\vv#1>{\vadjust{\vsk#1>}\ignorespaces}
\def\vvn#1>{\vadjust{\nobreak\vsk#1>\nobreak}\ignorespaces}

  \let\ssize\scriptstyle
\let\sssize\scriptscriptstyle

\let\Medskip\medskip
\def\medskip{\par\Medskip}
\let\Bigskip\bigskip
\def\bigskip{\par\Bigskip}

\let\Maketitle\maketitle
\def\maketitle{\Maketitle\thispagestyle{empty}\let\maketitle\empty}

\newtheorem{thm}{Theorem}[section]
\newtheorem{cor}[thm]{Corollary}
\newtheorem{lem}[thm]{Lemma}
\newtheorem{prop}[thm]{Proposition}

\newtheorem{defn}[thm]{Definition}
\newtheorem{ex}[thm]{Example}

\theoremstyle{definition}                                  
\numberwithin{equation}{section}

\theoremstyle{definition}
\newtheorem*{rem}{Remark}

\let\mc\mathcal
\let\nc\newcommand

\let\la\lambda

\let\phi\varphi

\let\der\partial

\let\geq\geqslant

\let\leq\leqslant

\let\on\operatorname
\let\bi\bibitem
\let\bs\boldsymbol

\def\C{{\mathbb C}}
\def\Z{{\mathbb Z}}

\def\res{\mathrm{res}}

\def\F{{\mathbb F}}   

\def\+#1{^{\{#1\}}}

\def\beq{\begin{equation}}
\def\eeq{\end{equation}}
\def\be{\begin{equation*}}
\def\ee{\end{equation*}}

\nc{\bea}{\begin{eqnarray*}}
\nc{\eea}{\end{eqnarray*}}
\nc{\bean}{\begin{eqnarray}}
\nc{\eean}{\end{eqnarray}}

\nc{\Il}{{\mc I_{\bs\la}}}
\nc{\bla}{{\bs\la}}
\nc{\Fla}{\F_\bla}
\nc{\tfl}{{T^*\Fla}}
\nc{\GL}{{GL_n(\C)}}
\nc{\GLC}{{GL_n(\C)\times\C^*}}

\let\sd s 

\def\ddk_#1{\kk_{#1}\<\>\frac\der{\der\<\>\kk_{#1}}}

\def\bul{\mathbin{\raise.2ex\hbox{$\sssize\bullet$}}}
\def\intt{\mathchoice
{\mathop{\raise.2ex\rlap{$\,\,\ssize\backslash$}{\intop}}\nolimits}
{\mathop{\raise.3ex\rlap{$\,\sssize\backslash$}{\intop}}\nolimits}
{\mathop{\raise.1ex\rlap{$\sssize\>\backslash$}{\intop}}\nolimits}
{\mathop{\rlap{$\sssize\<\>\backslash$}{\intop}}\nolimits}}

\let\kk q 
\let\cc c

\let\Ko K

\def\GZ/{Gelfand-Zetlin}
\def\KZ/{{\slshape KZ\/}}
\def\qKZ/{{\slshape qKZ\/}}
\def\XXX/{{\slshape XXX\/}}

\nc{\A}{{\mc A}}

\def\cb{{\mathcal{M}^{\times}_C}}

\nc{\hsl}{\widehat{{\frak{sl}_2}}}

\nc{\BC}{{ \mathbb C}}
\nc{\lra}{\longrightarrow}
\nc{\CO}{{\mathcal{O}}}
\nc{\BZ}{{ \mathbb Z}}
\nc{\hfn}{\hat{\frak{n}}}
\nc\Zs{{\Z/p^s\Z}}
\nc\Zo{{\Zs[z]^0}}
\nc\gr{{\on{gr}}}

\nc\fD{{\frak D}}

\usepackage{tikz}

\usetikzlibrary{decorations}
\usetikzlibrary{decorations.pathmorphing}
\usetikzlibrary{calc}

\begin{document}

\title[$A_1$ - Coulomb branches over finite fields and Selberg Character sums]
{$A_1$-Coulomb branches over finite fields and Selberg Character sums}

\author
[Junzhe Lyu and Andrey Smirnov]
{Junzhe Lyu$^{\diamond}$  and Andrey Smirnov$^{\star}$ }

\maketitle

\vspace{-2mm}

\begin{center}
{ Department of Mathematics, University
of North Carolina at Chapel Hill\\ Chapel Hill, NC 27599-3250, USA\/}
\end{center}

{\let\thefootnote\relax
\footnotetext{\vsk-.8>\noindent
$^\diamond\<${\sl E\>-mail}:\enspace  taiat@unc.edu
\\
$^\star\<${\sl E\>-mail}:\enspace asmirnov@unc.edu
}}

\begin{center}
\begin{abstract}
Motivated by mirror symmetry, we study exponential sums over the $\mathbb{F}_q$-points of the $A_1$ Coulomb branch. We show that their fiberwise structure is governed by polynomial Gauss sums, providing a geometric realization of Selberg character sums studied by Evans. Using Evans's evaluation, we obtain explicit formulas for the Coulomb-branch exponential sums as products of classical Gauss sums.
\end{abstract}
\end{center}

\section{Introduction}

\subsection{Mirror symmetry and Coulomb branches} 

In the mirror-symmetry framework of \cite{Aga1}, Aganagic and her collaborators study pairs $(\mathcal{M}^{\times}_C,W)$, where $\mathcal{M}^{\times}_C$ is a Poisson algebraic variety known as a multiplicative Coulomb branch and 
$
W : \mathcal{M}^{\times}_C \longrightarrow \mathbb{C}
$
is the {\it superpotential}. The superpotential has additive and multiplicative parts and takes the schematic form
$$
W= \alpha\,\log(W_m) + \beta\, W_a.
$$
The associated Landau - Ginzburg model is described by the $W$-twisted de Rham cohomology of $\mathcal{M}^{\times}_C$ and its Gauss-Manin connection. Its solutions are represented by oscillatory integrals
\bean \label{mirint}
I_{\gamma} =
\int\limits_{\gamma \subset \mathcal{M}^{\times}_C}
e^W\,\omega
=
\int\limits_{\gamma \subset \mathcal{M}^{\times}_C}
(W_m)^\alpha \exp(\beta W_a)\,\omega,
\eean
where $\omega$ is the canonical Liouville form and $\gamma$ is a Lagrangian cycle. Mirror symmetry identifies this Gauss-Manin system with the quantum connection on the Higgs side, so that the integrals (\ref{mirint}) give solutions of the corresponding quantum differential equations.

The purpose of this note is to develop an arithmetic counterpart of this picture, replacing oscillatory integrals on complex Coulomb branches by exponential character sums over their finite-field points. Already for the simplest $A_1$ Coulomb branches, these sums turn out to be classical objects of arithmetic: polynomial Gauss sums and Selberg character sums.

\subsection{Exponential sums associated with Coulomb branches}

The oscillatory integrals (\ref{mirint}) have natural arithmetic counterparts over finite fields. The power terms are replaced by multiplicative characters
$
\chi:\mathbb{F}_q^{\times}\longrightarrow{\mathbb{C}}^{\times},
$
while the exponential terms are replaced by additive characters
$
\psi:\mathbb{F}_q\longrightarrow {\mathbb{C}}^{\times}.
$
This leads to the exponential sums
\bean \label{expsum}
S=
\sum\limits_{x\in\cb(\mathbb{F}_q)}
\chi(W_m(x)) \psi( W_a(x)).
\eean
From the cohomological point of view, such sums arise as Frobenius traces via the Grothendieck--Lefschetz trace formula, providing a finite-field counterpart of the Gauss--Manin picture described above. Varying the finite-field extension $\mathbb{F}_{q^s}$ amounts to considering traces of $s$-powers of Frobenius. When $q$ is a power of a prime $p$, the resulting Frobenius action is expected to be the Dwork-type $p$-adic Frobenius structure on the corresponding quantum connection. We refer to \cite{Ked} for a general review of $p$-adic Frobenius structures on differential equations, to \cite{BLP,BPS,HL,BL2} for recent work on Frobenius structures in quantum cohomology, and to \cite{BL,Smi,KS} for related developments in quantum $K$-theory. The present work was motivated in part by our attempts to compute the traces of the Frobenius structures arising in this way.  

 We will see that Coulomb  branch geometry provides a natural framework in which classical finite-field character sums arise, while known identities for these sums in turn reveal arithmetic properties of the corresponding Coulomb branches.

\subsection{Coulomb exponential sums in type $A_1$}

We study the exponential sum (\ref{expsum}) for the Coulomb branch associated with the $A_1$ quiver. In this case the Coulomb branch is determined by the gauge group $G=GL(d)$ and its representation
$
N=Hom(\mathbb{C}^d,\mathbb{C}^n).
$
We refer to the case $n=0$ as the {\it pure gauge} Coulomb branch and to the case $n>0$ as the {\it deformed} Coulomb branch.

The deformation is determined by a polynomial
$$
f\in\mathbb{F}_q[x], \qquad \deg(f)=n.
$$
Over an algebraically closed field we have
$
f(u)=\prod_{i=1}^{n}(u-a_i),
$
and the roots $a_i$ play the role of equivariant parameters of the framing space $\mathbb{C}^n$. Over $\mathbb{F}_q$, however, it is natural to allow arbitrary polynomials $f$, which need not split into linear factors. We denote the corresponding $f$ - deformed Coulomb branch by $\mathcal{M}^{\times}_{f,d}(\mathbb{F}_q)$. The pure gauge case, corresponding to $f(u)=1$, is denoted by $\mathcal{M}^{\times}_{d}(\mathbb{F}_q)$.

\subsection{Character sums as polynomial Gauss sums}

The Coulomb branch comes with a natural projection
$$
\pi: \mathcal{M}^{\times}_{f,d}(\mathbb{F}_q)
\longrightarrow
B:=\mathrm{Spec}(\mathbb{F}_q[y_1^{\pm 1},y_2^{\pm 1},\dots,y_d^{\pm 1}]^{\mathrm{sym}}).
$$
The $\mathbb{F}_q$-points of the base $B$ can be identified with monic polynomials
$$
m \in\mathbb{F}_q[u], \qquad \deg(m)=d, \qquad m(0)\neq 0.
$$
In the pure gauge case, the points of the fiber $\pi^{-1}(m)$ are naturally identified with the units of the finite ring
$$
R_m=\mathbb{F}_q[x]/(m).
$$
In Section \ref{fibsumsec} we show that the characters $\chi$ and $\psi$ appearing in (\ref{expsum}) lift to multiplicative and additive characters
$
\widetilde{\chi}:R_m^{\times}\longrightarrow{\mathbb{C}}^{\times}$, $
\widetilde{\psi}:R_m\longrightarrow {\mathbb{C}}^{\times}.
$
Consequently, the restriction of (\ref{expsum}) to the fiber $\pi^{-1}(m)$ takes the form
$$
S(m)=\sum\limits_{a\in R_m^{\times}}
\widetilde{\chi}(a)\widetilde{\psi}(a),
$$
which is precisely a polynomial Gauss sum over the quotient ring $R_m$ \cite{Hay,Zhe}. Thus the natural projection $\pi$ decomposes the Coulomb branch exponential sum into classical finite ring Gauss sums.

Our first result gives an explicit evaluation of these fiberwise sums.

\begin{thm} \label{thm1}
We have
$$
S(m)=
\chi(-1)^{\frac{d(d-1)}{2}}
\phi(D)\,
G_q(\chi,\psi)^d\,
\chi(D),
$$
where $D$ denotes the discriminant of $m$ (we assume that $\chi(0)=0$), $\phi$ is the quadratic character of $\mathbb{F}_q^{\times}$, and
$$
G_q(\chi,\psi)=
\sum\limits_{x\in\mathbb{F}_q^{\times}}
\chi(x)\psi(x)
$$
is the classical $\mathbb{F}_q$ - Gauss sum.
\end{thm}

The deformed Coulomb branch $\mathcal{M}^{\times}_{f,d}$ and the pure gauge Coulomb branch $\mathcal{M}^{\times}_{d}$ are isomorphic away from the resultant locus $\res(f,m)=0$ in $B$. This allows us to extend the preceding formula to the deformed case $f\neq 1$:

\begin{thm} \label{thm2}[Theorem \ref{thmfibdef} for trivial $\chi'$]
We have
$$
S(m)=
\chi(-1)^{\frac{d(d-1)}{2}}
\phi(D)\,
G_q(\chi,\psi)^d\,
\chi(D)\chi^{-1}(\res(m,f)),
$$
where $\res(m,f)$ denotes the resultant of $m$ and $f$, and we assume that $\chi(0)=0$.
\end{thm}

The appearance of the discriminant and resultant in these formulas provides a direct link between the geometry of the Coulomb branch and classical polynomial character sums.

\subsection{Selberg character sums}

A particularly interesting specialization occurs when the deformation polynomial is supported at two points. Assume that
$$
f(u)=u^\alpha(u-1)^\beta,
$$
and let $\tau$ be a generator of the group of multiplicative characters. Up to an irrelevant sign, Theorem \ref{thm2} gives the following expression for the sum over the fiber $\pi^{-1}(m)$:
$$
S(m)=
G_q(\tau^c,\psi)^d\,
\phi(D)\,
\tau(D^c\,m(0)^a m(1)^b).
$$
for some $a,b,c \in \{0,\dots, q-1\}$.
Summing over the base, the full exponential sum (\ref{expsum}) becomes
$$
S=
G_q(\tau^c,\psi)^d
\sum_{m}
\phi(D)\,
\tau(D^c\,m(0)^a m(1)^b),
$$
where the sum is over monic polynomials $m$ of degree $d$. The remaining sum is precisely the Selberg character sum studied by Evans \cite{Eva}. Thus the Selberg sum arises naturally as the base sum of the fiberwise Gauss sums associated with the Coulomb branch.

The main result of \cite{Eva} gives an explicit evaluation of the Selberg sum as a product of classical $\mathbb{F}_q$-Gauss sums. Combined with the fiberwise calculation above, it gives the following closed formula for the Coulomb-branch exponential sum.

\begin{thm}
The exponential sum over the Coulomb branch $\mathcal{M}^{\times}_{f,d}(\mathbb{F}_q)$ equals
$$
S=
\prod\limits_{j=0}^{d-1}
\dfrac{
G_q(\tau^{a+jc},\psi)
G_q(\tau^{b+jc},\psi)
G_q(\tau^{c+jc},\psi)
}{
G_q(\tau^{a+b+(d-1+j)c},\psi)
}.
$$
\end{thm}

This formula also makes the behavior under finite-field extensions particularly transparent. Let
$$
\tau^{(s)}=\tau\circ N:
\mathbb{F}_{q^s}^{\times}\longrightarrow\mathbb{C}^{\times},
\qquad
\psi^{(s)}=\psi\circ\textrm{tr}:
\mathbb{F}_{q^s}\longrightarrow\mathbb{C}^{\times}
$$
be the lifted characters, and let $S_{q^s}$ denote the corresponding exponential sum over $\mathcal{M}^{\times}_{f,d}(\mathbb{F}_{q^s})$. The Davenport--Hasse lifting formula
$$
G_{q^s}(\tau^{(s)},\psi^{(s)})
=
(-1)^{s-1}G_q(\tau,\psi)^s
$$
implies, after cancellation of the signs in the product above, that
$$
S_{q^s}=(S_q)^s.
$$
We view this identity as a Davenport-Hasse lifting theorem for Coulomb-branch character sums. Moreover, since a nontrivial Gauss sum satisfies
$
|G_q(\tau,\psi)|=\sqrt{q},
$
the product formula gives
$$
|S_q|=q^{d/2}.
$$
Together with the relation $S_{q^s}=(S_q)^s$, this suggests a simple cohomological interpretation. Namely, the Kummer-Artin-Schreier local system on $\mathcal{M}^{\times}_{f,d}$ associated with the characters $\chi$ and $\psi$ should behave, at the level of Frobenius traces, as a one-dimensional pure object of weight $d$. It would be interesting to understand this directly from the geometry of $\mathcal{M}^{\times}_{f,d}$.

\subsection{Counting rational points}

In addition to the twisted point counts above, we study the ordinary point counts of the same Coulomb branches. Let
$
|\mathcal{M}^{\times}_{f,d}(\mathbb{F}_q)|
$
denote the number of $\mathbb{F}_q$-points of the deformed Coulomb branch and define the generating function
$$
\mathcal{Z}_f(z)
=
\sum\limits_{d=0}^{\infty}
|\mathcal{M}^{\times}_{f,d}(\mathbb{F}_q)|z^d.
$$
In Section \ref{pointcountse} we obtain the following explicit formula.

\begin{thm}
We have
$$
\mathcal{Z}_f(z)=
\dfrac{(1-zq)^2}{(1-z)(1-zq^2)}
\prod\limits_{p\mid f,\atop p\neq x}
\dfrac{
1-(qz)^{(\nu_p(f)+1)\deg(p)}
}{
1-(qz)^{\deg(p)}
},
$$
where the product runs over irreducible polynomials $p\in\mathbb{F}_q[x]$, $p\neq x$, dividing $f$, and $\nu_p(f)$ denotes the corresponding valuation.
\end{thm}
In the pure gauge case $f=1$, this reduces to
$$
\mathcal{Z}(z)
=
\dfrac{(1-zq)^2}{(1-z)(1-zq^2)}.
$$
Thus both the twisted and untwisted point counts of these $A_1$ Coulomb branches admit simple explicit descriptions: the former in terms of classical Gauss  and the latter through a simple rational generating function.

\section*{Acknowledgments} 
We thank Kifung Chan, Yixuan Li, Spencer Tamagni, Peng Zhou for valuable discussion about Coulomb branches and monopole scattering matrices. 
Work of A. Smirnov is partially supported by NSF grant
DMS-2401380 and by the Simons Foundation grant “Travel Support for Mathematicians”.


\section{Coulomb branches} 
\subsection{Coulomb branches}
Let $G$ be a complex reductive group and $\mathcal{G}r_G:=G(\mathcal{K})/G(\mathcal{O})$ the associated affine Grassmannian, where $\mathcal{K}=\mathbb{C}((z))$ is the field of formal Laurent series and $\mathcal{O}=\mathbb{C}[[z]]$ is the ring of formal power series.
It was shown in \cite{BFM} that $\mathrm{K}^{G(\mathcal{O})}(\mathcal{G}r_G)$ carries a convolution product that makes it into a commutative algebra. The corresponding affine scheme
$\textrm{Spec}(\mathrm{K}^{G(\mathcal{O})}(\mathcal{G}r_G))$
is called the pure gauge multiplicative Coulomb branch.


A more general construction, the so-called BFN Coulomb branch, was introduced in \cite{BFN}. The construction takes as input a complex reductive group $G$ and a representation $N$ of $G$. One defines
\[
\mathcal{T}_{G,N}:=G(\mathcal{K})\times_{G(\mathcal{O})} N(\mathcal{O}),
\]
and its subspace $ \mathcal{R}_{G,N}\subset \mathcal{T}_{G,N}$:
\[
\mathcal{R}_{G,N}:=\{(g,s)\in \mathcal{T}_{G,N}|\quad gs\in N(\mathcal{O})\}.
\] 
The equivariant $K$-theory $\mathrm{K}^{G(\mathcal{O})}(\mathcal{R}_{G,N})$ carries a convolution product, making it into a commutative algebra. The BFN multiplicative Coulomb branch associated with the pair $(G,N)$ is then defined as
$
\mathrm{Spec}(\mathrm{K}^{G(\mathcal{O})}(\mathcal{R}_{G,N})).
$


If there exists a torus $T$ acting on $N$ such that the $T$ action and the $G$ action commute, then the Coulomb branch admits a deformation defined as $\mathrm{Spec}(\mathrm{K}^{G(\mathcal{O})\times T}(\mathcal{R}_{G,N}))$. By turning off all equivariant parameters of $T$, we recover the ordinary Coulomb branch as the central fiber of the deformation family.

Coulomb branches come with a natural flat morphism
\bean \label{propi}
\pi: \mathrm{Spec}(\mathrm{K}^{G(\mathcal{O})}(\mathcal{R}_{G,N}))\to B:= \mathrm{Spec}(\mathrm{K}^{G(\mathcal{O})}(pt)),
\eean 
induced by the natural homomorphism
$
\mathrm{K}^{G(\mathcal{O})}(pt)\to \mathrm{K}^{G(\mathcal{O})}(\mathcal{R}_{G,N}).
$

\subsection{Monopole operators}
$\mathcal{G}r_G$ has a natural $G_{\mathcal{O}}$ action from the left, which induces the following Bruhat-type decomposition
\[
\mathcal{G}r_G=\coprod\limits_{\lambda\in \chi(T)^{+}}\, G_{\mathcal{O}}t^{\lambda},
\]
where $\chi(T)^{+}$ is the cone of dominant coweights. Every open cell $G_{\mathcal{O}}t^{\lambda}$ is a smooth variety, and $\overline{G_{\mathcal{O}}t^{\lambda}}=\sqcup_{\mu\leq \lambda} G_{\mathcal{O}}t^{\mu}$ defines a closed subvariety of $\mathcal{G}r_G$ for all $\lambda\in \chi(T)^+$. In particular, if $\lambda$ is minuscule, then $G_{\mathcal{O}}t^{\lambda}$ is smooth.

Let $p: \mathcal{R}_{G,N}\to \mathcal{G}r_G=G_{\mathcal{K}}/G_{\mathcal{O}}$ be the natural projection defined by $p:(g,s)\to g$. 
\begin{defn}
For minuscule $\lambda$, we call the class $[\mathcal{O}_{p^{-1}(G_{\mathcal{O}}t^{\lambda})}] \in \mathrm{K}^{G(\mathcal{O})}(\mathcal{R}_{G,N}) $ the monopole operator.
\end{defn}
More generally,
\begin{defn}
For any $G_{\mathcal{O}}$-equivariant coherent sheaf $\mathcal{F}$ over $G_{\mathcal{O}}t^{\lambda}$, we call the class $[p^{*}\mathcal{F}]\in \mathrm{K}^{G(\mathcal{O})}(\mathcal{R}_{G,N})$ a dressed monopole operator.
\end{defn}
Since $\mathcal{R}_{G,N}|_{G_{\mathcal{O}}t^{\lambda}}\to G_{\mathcal{O}}t^{\lambda}$ is a vector bundle (even for non-minuscule $\lambda$), there is no need to distinguish between the ordinary and derived pullbacks. 
\begin{lem}[\cite{Web}, Proposition 3.1;\cite{VV}, Remark 4.2.6]
The Coulomb branch algebra $\mathrm{K}^{G(\mathcal{O})}(\mathcal{R}_{G,N})$ is generated by dressed monopole operators and $K^{G_{\mathcal{O}}}(pt)$.
\end{lem}


\subsection{Coulomb branch of type $A_1$}
The Coulomb branch we consider in this note is the Coulomb branch associated with a representation of the $A_1$ quiver, i.e., we take $G=GL(d, \mathbb{C})$ and $N=\mathrm{Hom}(\mathbb{C}^d, \mathbb{C}^n)$, where we understand $\mathbb{C}^d$ as the standard fundamental representation of $G$. In this case, one has the flavor symmetry $T=\prod_{i=1}^n\mathbb{C}_{a_i}^*$ acting on the framing $\mathbb{C}^n$, providing a deformation family of the Coulomb branch associated with a representation of the $A_1$ quiver.

For $G=GL(d)$, $\lambda=(1,...,1,0,...,0)$ or $(0,...,0,-1,...,-1)$ are all minuscule. 
\begin{ex}
The minuscule $G(\mathcal{O})$ orbit $G(\mathcal{O})t^{(0,...,0,-1)}$ is isomorphic to the Grassmannian~$Gr(1,d)$, and the minuscule $G(\mathcal{O})$ orbit $G(\mathcal{O})t^{(1,0,...,0)}$ is isomorphic to $Gr(d-1,d)$.
\end{ex}

Denote the tautological quotient bundle of rank one over the minuscule $G(\mathcal{O})$ orbit $G(\mathcal{O})t^{(1,0,...,0)}\cong Gr(d-1,d)$ by $Q$, and the tautological bundle of rank one over the minuscule $G(\mathcal{O})$ orbit $G(\mathcal{O})t^{(0,...,0,0,-1)}\cong Gr(1,d)$ by $S$. The base of the projection in this case is
 $$
B=Spec(K^{GL(d,\mathcal{O})}(pt)) \cong  Spec(\mathbb{C}[y_1^{\pm 1},y_2^{\pm 1},\dots, y_d^{\pm 1}]^{sym})
 $$
 where $sym$ denotes symmetric polynomials. 
 It will be convenient to identify $B$ with the space of monic polynomials $m(u)$ of degree $\deg m(u)=d$ with non-vanishing constant term (so that $y_i$ are the roots of $m(u)$).

 In the rest of this section, we are going to prove the following.
\begin{thm} \label{genthm}
The monopole operators $[(p^*Q)^{\otimes p}]$, $[(p^*S)^{\otimes p}]$ for $p\in \{0,1,...,d-1\}$ together with $K^{GL(d,\mathcal{O})}(pt)$ generate the $A_1$ Coulomb branch algebra.
\end{thm}
\begin{proof}
Equivalently, we can prove the classes $[(p^*(\mathcal{O}^{\oplus d}-Q))^{\otimes p}]$ and $[(p^*S)^{\otimes p}]$ are set of generators.

First, one can define an embedding:
\[(p^*)^{-1}i_*:K^{GL(d,\mathcal{O})}(R_{GL(d,\mathbb{C}),N})\otimes_{\mathbb{Z}}\mathbb{C}\to K^{GL(d,\mathcal{O})}(\mathcal{G}r_{GL(d,\mathbb{C})})\otimes_{\mathbb{Z}}\mathbb{C}.
\] 
Here \[i: R_{G,N}\to T_{G,N}\] the closed embedding and \[p: T_{G,N}\to Gr_G\] a vector bundle.

Now we work on the locus that $\prod_{i\neq j}(1-\frac{y_i}{y_j})\neq 0$. After localizing $\prod_{i\neq j}(1-\frac{y_i}{y_j})$ we have:
\[
(K^{GL(d,\mathcal{O})}(\mathcal{G}r_{GL(d,\mathbb{C})})\otimes_{\mathbb{Z}}\mathbb{C})_{loc}\cong (K^{\mathbb{C}^{\times}(\mathcal{O})^d}(\mathcal{G}r_{\mathbb{C}^{\times})^d}) \otimes_{\mathbb{Z}}\mathbb{C})_{loc}^{S_d}.
\]
Combine these two maps together; we have the following embedding:
\[
 (K^{GL(d,\mathcal{O})}(R_{GL(d,\mathbb{C}),N})\otimes_{\mathbb{Z}}\mathbb{C})_{loc}\to (K^{\mathbb{C}^{\times}(\mathcal{O})^d}(\mathcal{G}r_{\mathbb{C}^{\times})^d}) \otimes_{\mathbb{Z}}\mathbb{C})_{loc}^{S_d}.
\]
Under such embedding, the image of $z_p$ and $w_p$ can be written as 
\[
z_{p}:=[(p^*(\mathcal{O}^{\oplus d}-Q))^{\otimes p}]=\sum_i \frac{y^{p}_ix_i}{\prod_{j\neq i}(1-\frac{y_j}{y_i})}, \ \ \ 
w_p:= [S^{\otimes p}]=\sum_i \frac{y^p_i\prod_{k=1}^{n}(1-\frac{a_k}{y_i})x_i^{-1}}{\prod_{j\neq i}(1-\frac{y_i}{y_j})}.
\]
A detailed computation can be found in section 2 of \cite{SS}. Where they turn on the loop rotation but we do not.

The localization formula can be rewritten in terms of matrix form.
\[
\langle z_{-d+1},...,z_0\rangle^{T}=V\langle x_1,...,x_d \rangle^{T}
\]

Then we can solve out $x_i$ as
\[
x_i=V^{-1}_i\langle z_{-d+1},...,z_{0}\rangle^T
\]
Here the inverse matrix $V^{-1}$ is
\[
V^{-1}=\begin{bmatrix}
(-1)^{d-1}\prod_{j\neq 1}y_j & (-1)^{d-
2}\sum_i \prod_{j\neq i,1} y_j & ... & 1 \\
(-1)^{d-1}\prod_{j\neq 2}y_j& (-1)^{d-
2}\sum_i \prod_{j\neq i,2} y_j & ... & 1 \\
... & ... & ...&... \\
(-1)^{d-1}\prod_{j\neq d}y_j & (-1)^{d-
2}\sum_i \prod_{j\neq i,d} y_j & ... & 1\\
\end{bmatrix}
\]
Which can be written as
\[
V^{-1}=\begin{bmatrix}
(-1)^{d-1}\sum_{k=1}(-y_1)^{k-1}e_{d-k}(y) & (-1)^{d-
2}\sum_{k=2}(-y_1)^{k-2}e_{d-k}(y) & ... & 1 \\
(-1)^{d-1}\sum_{k=1}(-y_2)^{k-1}e_{d-k}(y) & (-1)^{d-
2}\sum_{k=2}(-y_2)^{k-2}e_{d-k}(y) & ... & 1\\
... & ... & ...&... \\
((-1)^{d-1}\sum_{k=1}(-y_d)^{k-1}e_{d-k}(y) & (-1)^{d-
2}\sum_{k=2}(-y_d)^{k-2}e_{d-k}(y) & ... & 1\\
\end{bmatrix}
\]
Under such rewriting we have,
\[
x_i=\sum_{k\geq 1}^{d}(\sum_{n\geq 0}^{k-1}(-1)^ne_n(y)z_{-d+k-n})y_j^{d-k}.
\]
Given an dressed positive minuscule monopole operator $R$ with weight $(1,...,1,0,...,0)$ and  $|(1,...,1,0,...,0)|=k$, its image under localization is
\[
R=\sum_{w\in S_d/(S_k\times S_{d-k})} \frac{f(w\cdot y_i)\prod_{i\in wI} y_i^{d-k}x_i}{\prod_{i\in wI,j\notin wI} (y_i-y_j)}
\]
Here $I=\{1,...,k\}$, $f(y_i)$ is a symmetric polynomial in variables $y_i, i\in I$,

$x_j=\sum_{k\geq 1}^{d}(\sum_{n\geq 0}^{k-1}(-1)^ne_n(y)z_{-d+k-n})y_j^{d-k}$.

Such a fractional expression may have simple poles at $y_a=y_b$ for $a\neq b$. Now we are going to prove it indeed has no pole, i.e. the residue of this fractional expression is zero at $y_a=y_b$ for any $a\neq b$.

Denote $S=\{1,...,d\}$ the total index set. Fix $a\neq b$ for some $a, b\in S$, then any size $k$ subset $a\in A\subset S$ naturally pairs up with another size $k$ subset $b\in B\subset S$, such that 
\[
B\backslash\{b\}=A\backslash\{a\}.
\]Since 
\[
\prod_{i\in A,j\notin A} (y_i-y_j)=(y_a-y_b)\prod_{j\in S\backslash(A\cup B)} (y_a-y_j)\prod_{i\in A\backslash\{a\}} (y_i-y_b)\prod_{i\in A\backslash\{a\}, j\in S\backslash(A\cup B)} (y_i-y_j).
\]
And
\[
\prod_{i\in B,j\notin B} (y_i-y_j)=(y_b-y_a)\prod_{j\in S\backslash(A\cup B)} (y_b-y_j)\prod_{i\in B\backslash\{b\}} (y_i-y_a)\prod_{i\in B\backslash\{b\}, j\in S\backslash(A\cup B)} (y_i-y_j).
\] 
One can compute the Residue
\[
\lim_{y_a-y_b\to 0}(y_a-y_b)(\frac{f(y_i)\prod_{i\in A} y_i^{d-k}x_i}{\prod_{i\in A,j\notin A} (y_i-y_j)}+
\frac{f(y_i)\prod_{i\in B} y_i^{d-k}x_i}{\prod_{i\in B,j\notin B} (y_i-y_j)})
\]
\[
=\lim_{y_a-y_b\to 0}((y_a-y_b)\frac{f(y_i)\prod_{i\in A} y_i^{d-k}x_i}{\prod_{i\in A,j\notin A} (y_i-y_j)}-(y_b-y_a)
\frac{f(y_i)\prod_{i\in B} y_i^{d-k}x_i}{\prod_{i\in B,j\notin B} (y_i-y_j)})=0
\]

By running through all possible $A$ we conclude that this formula has no pole at $y_a-y_b$.

$X$ is a symmetric rational function in $y_i$, and the residue calculation above implies that this is indeed a symmetric polynomial in $y_i$. Since any symmetric polynomial in $y_i$ can be expressed via elementary symmetric polynomials $e_k(y_i)$, we expressed any dressed monopole operators associated to positive coweight as a $\mathbb{Z}$ combination of monomials in $z_p$ and $e_k(y_i)$ on the locus $y_i$ are distinct. This identification can be extended to the entire space, hence we are done.
Therefore any dressed minuscule monopole operators with positive weight can be generated by $z_p$ and $\mathrm{K}^{GL(d,\mathcal{O})}(pt)$.

For negative dressed minuscule operators, we introduce the following change of variables
\[s_i=\prod_{k\neq i}^{n}(1-\frac{a_k}{y_i})x_i^{-1}\]
to rewrite $w_p$ as,
\[
w_p:= [S^{\otimes p}]=\sum_i \frac{y^p_is_i}{\prod_{j\neq i}(1-\frac{y_i}{y_j})}
\]
Notice under localization, any negative dressed monopole operators will be fractional combination of $s_i$ and $y_i$. Similar argument shows any dressed minuscule monopole operators associated to negative coweight can be generated by $w_p$ and $\mathrm{K}^{GL(d,\mathcal{O})}(pt)$, one just simply solves out $s_i$ and expresses all dressed monopole operators as a combination of $s_i$ and $y_i$.
\end{proof}
\begin{rem}
We remark that this proof also works for the $\mathbb{F}_q-$ coefficient $\mathrm{K}-$ theory for prime power $q$. To use the localization theorem, one needs to base change to the algebraic closure $\bar{\mathbb{F}}_q$, to ensure the Euler class $\prod_{i\in wI, j\notin wI}(1-y_i/y_j)$ is not everywhere zero. Then one obtains the same expression of dressed monopole operators as $\mathbb{Z}-$ expansion of $z_p, w_p, e_k(y_i)$, finally one can conclude the same equality holds over $\mathbb{F}_q$ via Galois descent. 
\end{rem}
\begin{rem}
For the $A_1$ quiver gauge theory without framing, \cite{CW} showed that the $K$ theoretic Coulomb branch is generated by dressed monopole operators with weights $(1,0,...,0)$ and $(0,...,0,-1)$ under $\mathbb{Z}$ coefficient.
\end{rem}


\section{Realizations of $A_1$ Coulomb branch} 
\subsection{Generators and relations}
Let us consider the following monopole operators 
\bean \label{abeliz}
b_k =(-1)^k \sum\limits_{i=1}^{d}\, x_i \dfrac{e_{k}(\bar{Q}_i)}{\prod\limits_{j\neq i}(1-\frac{y_i}{y_j}) }, \ \ \ 
a_k =(-1)^k \sum\limits_{i=1}^{d}\, \frac{f(y_i)}{x_i} \dfrac{e_{k}(\bar{Q}_i)}{\prod\limits_{j\neq i}(1-\frac{y_i}{y_j}) }, 
\eean
$k=0,\dots,d-1$, where $e_{k}(\bar{Q}_i)$ denotes the $k$-th elementary symmetric function in the $d-1$ variables $1/y_s$, $s\in \{1,\dots, d\} \setminus \{i\}$. Notice $a_k$ is just localization of the class $[p^*(\mathcal{O}^{\oplus d}-S)^{\otimes k}]$, by Theorem \ref{genthm}, the coordinate ring of the Coulomb branch is the algebra generated by $3d$ elements 
\bean \label{genrs}
a_0,\dots, a_{d-1}, b_0,\dots, b_{d-1}, e_1,\dots, e_d,
\eean 
which are subject to $d$ relations (since the known dimension is $2d$). To describe these relations, it is convenient to introduce the corresponding generating functions 
\bean \label{abmcoef}
a(u) = \sum\limits_{i=0}^{d-1}\, a_i u^i,  \ \ \  b(u) = \sum\limits_{i=0}^{d-1}\, b_i u^i, \ \ \ m(u) = \sum\limits_{i=0}^{d}\, (-1)^i e_i u^{d-i}
\eean 
Explicitly, we have
$$
b(u) =\sum\limits_{i=1}^{d}\, x_i \prod\limits_{j\neq i} \frac{(1-\frac{u}{y_j})}{(1-\frac{y_i}{y_j})}, \ \ \ a(u) =\sum\limits_{i=1}^{d}\, \frac{f(y_i)}{x_i} \prod\limits_{j\neq i} \frac{(1-\frac{u}{y_j})}{(1-\frac{y_i}{y_j})}, \ \ \ m(u) =\prod\limits_{i=1}^{d}(u-y_i). 
$$
Restricting $a(u)$ and $b(u)$ to the roots of $m(u)$, we obtain
$$
b(y_i) = x_i, \ \ a(y_i) =\frac{f(y_i)}{x_i}.
$$
This means that 
$
a(u) b(u) \equiv  f(u) \pmod{m(u)}. 
$
This congruence is equivalent to $d$ equations for the generators (\ref{genrs}). Indeed, let us denote
\bean \label{rmring}
R_{m(u)} = \mathbb{F}_q[u]/(m(u)).
\eean 
Since $m(u)$ is a monic polynomial of degree $d$, $R_{m(u)}$ is a $d$-dimensional vector space over $\mathbb{F}_q$. The image of the above congruence in $R_{m(u)}$ is an equality of two vectors
$a b = f.$
Fixing any $\mathbb{F}_q$-basis of $R_{m(u)}$ gives $d$ equations for the coefficients of these vectors.

\begin{prop} \label{proprel} 
The $A_1$ multiplicative Coulomb branch is an affine variety 
$$
\mathcal{M}^{\times}_{f,d} =\{a,b,m \in \mathbb{F}_q[u]: \deg(a) < d,\, \deg(b) < d, \deg(m) = d, a b \equiv\!\!f\pmod{m}\}
$$
where $m$ is assumed to be monic with $m(0)\neq 0$. 
\end{prop}
Let us recall that degree $d$ polynomials
$
m(u) = \sum\limits_{i=0}^{d}\, (-1)^i e_i u^{d-i}
$
parametrize the points in the base of the projection (\ref{propi}). 
\begin{cor}
The fiber of the projection map (\ref{propi}) over $m(u) \in \mathbb{F}_q[u]$ has the following form
$$
\pi^{-1}(m(u)) = \{(a,b) \in R^{2}_{m(u)}: a b =\bar{f}\}
$$
where $\bar{f}$ is the image of $f(u)$ in the quotient ring (\ref{rmring}). 
\end{cor}
The case of the pure gauge Coulomb branch $f=1$ is particularly simple. 
\begin{cor} \label{fiberpute}
The fiber $\pi^{-1}(m(u))$ of the pure gauge Coulomb branch consists of invertible elements $R^{\times}_{m(u)}$. 
\end{cor}

\subsection{Monopole scattering matrix}
The description of the Coulomb branch in Proposition \ref{proprel} is equivalent to the description via monopole scattering \cite{BDG,Tam}. The congruence in Proposition \ref{proprel} gives
\bean \label{concor}
a(u) b(u)  = f(u) + q(u) m(u).
\eean 
This equation can be written as a determinant:
\begin{prop}
The Coulomb branch $\mathcal{M}^{\times}_{f,d}(\mathbb{F}_q)$ is the space of matrices 
$$
\mathcal{S}(u) =\left(\begin{array}{cc}
m(u), & -b(u)\\
a(u), & -q(u)
\end{array}\right) \in \textrm{Mat}_{2}(\mathbb{F}_q[u]),
$$
such that $m(u)$ is monic of degree $d$ with non-vanishing constant term, $\deg(a(u))< d$, $\deg(b(u))< d$ and 
$\det \mathcal{S}(u) = f(u)$.
\end{prop}

It follows from (\ref{concor}) that
$$
\deg q(u)  \ \ \textrm{is}  \ \ \left\{\begin{array}{cc}
    \leq d-2, & \deg f(u) < 2 d-1,  \\
    =\deg f(u) -d, & \deg f(u)\geq 2 d-1.  
\end{array}\right.
$$
and the top coefficient of $q(u)$ is minus the top coefficient of $f(u)$ in the second case. We also note
$$
\mathcal{S}(u) = L(u) H(u) U(u) 
$$
where 
$$
L(u) = \left(\begin{array}{cc}
1, & 0\\
\dfrac{a(u)}{m(u)}, &1 
\end{array}\right), \ \ \ U(u) = \left(\begin{array}{cc}
1, & -\dfrac{b(u)}{m(u)}\\
0, &1 
\end{array}\right), \ \ \  H(u) = \left(\begin{array}{cc}
m(u), & 0\\
0, &\dfrac{f(u)}{m(u)} 
\end{array}\right).
$$
Since $\deg m(u)> \deg a(u)$ and $\deg m(u)> \deg b(u)$, as $u\to \infty$ we have
$$
L(u) \longrightarrow \left(\begin{array}{cc}
1 & 0\\
0 &1 
\end{array}\right), \ \ \ U(u) \longrightarrow \left(\begin{array}{cc}
1 & 0\\
0 &1 
\end{array}\right).
$$
A matrix $\mathcal{S}(u)$ satisfying these properties is called the monopole scattering matrix. The space of such matrices is referred to as the monopole moduli space. This gives an identification of the Coulomb branch with the moduli space of monopoles.

\section{Selberg integrals and character sums}
\subsection{Superpotential and mirror integrals}
For the $A_1$ Coulomb branch $\mathcal{M}^{\times}_{f,d}$, the superpotential, see formula (B4) in Section 2.3 of \cite{Aga2}, is given by:
$$
W = \mu \sum_{i=1}^{d} \, \log(x_i) + \sum\limits_{i=1}^{d} \dfrac{1}{x_i} \dfrac{f(y_i)}{\prod_{j\neq i} (y_i-y_j)},
$$
and the corresponding mirror integrals have the following form:
$$
I_{\gamma}=\int\limits_{\gamma} e^{W} d\omega  =\int\limits_{\gamma} \, (x_1\cdots x_d)^{\mu} \exp\Big(  \sum\limits_{i=1}^{d} \dfrac{1}{x_i} \dfrac{f(y_i)}{\prod_{j\neq i} (y_i-y_j)}  \Big) d\omega
$$
Using the presentation from Proposition \ref{proprel}, we may rewrite the integrand as:
$$
x_1\cdots x_d = \res(m(u),b(u)), \ \ \ \sum\limits_{i=1}^{d} \dfrac{1}{x_i} \dfrac{f(y_i)}{\prod_{j\neq i} (y_i-y_j)}  = \lambda(a(u)),
$$
where $\mathrm{res}(\cdot,\cdot)$ denotes the resultant of two polynomials and $\lambda(a(u)) =a_{d-1}$ is the top coefficient of the polynomial $a(u)$. In these coordinates, the mirror integrals have the following form:
\bean \label{mirint2}
I_{\gamma}=\int\limits_{\gamma} \, \res(m(u),b(u))^{\mu} e^{\lambda(a(u))} d\omega.
\eean

\subsection{Exponential sums} 
Let 
$$
\chi: \mathbb{F}_q^{\times } \rightarrow \mathbb{C}^{\times}, \ \ \ \psi: \mathbb{F}_q \rightarrow \mathbb{C}^{\times}
$$
be fixed multiplicative and additive characters. Let $t \in \mathcal{M}^{\times}_{f,d}(\mathbb{F}_q)$
be a point represented by polynomials $a(u),b(u),m(u)$ as in Proposition \ref{proprel}. We define
\bean \label{superpotdef}
\Phi:  \mathcal{M}^{\times}_{f,d}(\mathbb{F}_q) \longrightarrow \mathbb{C}^{\times}, \ \ \ \Phi: t \mapsto  \chi( \mathrm{res}(m(u),b(u)))^{-1} \psi( \lambda( a(u)) ). 
\eean 
The natural finite-field counterparts of the mirror integrals (\ref{mirint2}) are the finite sums
\bean \label{mirsum}
S = \sum\limits_{ t \in \mathcal{M}^{\times}_{f,d}(\mathbb{F}_q)}\, \Phi(t)
\eean

\subsection{Finite sums over fibers of $\pi$: pure gauge case \label{fibsumsec}}
Let
\bean \label{fqGS}
G_q(\chi,\psi)  = \sum_{x \in \mathbb{F}_q^{\times}}\, \chi(x) \psi(x)
\eean
be the $\mathbb{F}_q$-Gauss sum associated with the characters $\chi$ and $\psi$. The field $\mathbb{F}_q$ may be replaced by any commutative ring $R$, in which case we may speak of $R$-Gauss sums.

Let us note that in the pure gauge case, the sum (\ref{mirsum}) for fixed $m(u)$ (i.e., the sum over the points in the fiber $\pi^{-1}(m(u))$) is nothing but an example of a Gauss sum for the ring (\ref{rmring}). Indeed, by Corollary \ref{fiberpute}, the points in the fiber $\pi^{-1}(m(u))$ correspond to elements of $R^{\times}_{m(u)}$. Also, we have
 $$
 \mathrm{res}(m(u),b(u) b'(u)) = \mathrm{res}(m(u),b(u)) \mathrm{res}(m(u),b'(u)), \ \ \ \lambda(a(u) + a'(u) ) = a_{d-1}+a'_{d-1}.
 $$
Therefore, the maps 
\bean \label{charlifts}
\tilde \chi = \chi \circ  \mathrm{res}(m(u),\cdot): R^{\times}_{m(u)} \to \mathbb{C}^{\times}, \ \ \ \tilde{\psi}= \psi \circ \lambda : R_{m(u)} \to \mathbb{C}^{\times},
\eean
are the multiplicative and additive characters of $R_{m(u)}$. The sum (\ref{mirsum}) restricted to $\pi^{-1}(m(u))$ takes the following form:
$$
S(m(u))=\sum_{t \in \pi^{-1}(m(u))}\, \Phi(t)=\sum\limits_{a \in R^{\times}_{m(u)}} \, \tilde{\chi}(b)^{-1}  \tilde{\psi}(a).
$$
In the pure gauge case, $a b =1$, and thus 
\bean \label{sumw}
S(m(u))=\sum\limits_{a \in R^{\times}_{m(u)}} \, \tilde{\chi}(a)  \tilde{\psi}(a), 
\eean 
This is exactly the polynomial Gauss sum as defined in \cite{Hay}; see also \cite{Zhe} for more recent developments. This sum can be evaluated explicitly: 
\begin{thm} \label{ndth}
The finite sum (\ref{sumw}) equals
$$
S(m(u))=  \chi(-1)^{\frac{d(d-1)}{2}} \phi(D)\,  G_q(\chi,\psi)^d \chi(D)
$$
where $D$ is the discriminant of $m(u)$ (we extend $\chi(0)=0$), $G_q(\chi,\psi)$ is the $\mathbb{F}_q$-Gauss sum (\ref{fqGS}), and $\phi$ is the quadratic character of $\mathbb{F}^{\times}_q$:
$
\phi(D) = \left\{\begin{array}{ll}
+1, & \mathrm{if} \, D \, \textrm{is \,full\, square \,in\,}\, \mathbb{F}_q,\\
-1, & \textrm{else}.
\end{array}\right. 
$
\end{thm}
\begin{proof}
First, let us consider the case when $m(u)$ is square-free:
$$
m(u)=p_1(u) p_2(u) \cdots p_r(u),
$$
where $p_i(u)$ are mutually coprime irreducible polynomials. For each $p_i(u)$, we fix one of its roots $\alpha_{i} \in \bar{\mathbb{F}}_q$. By the Chinese remainder theorem,
 $$
 R= \mathbb{F}_q[u]/(m(u)) \cong R_1 \times \dots \times R_r,
 $$
where $R_i = \mathbb{F}[u]/(p_i(u)) \cong \mathbb{F}_{q^{d_i}}$. We fix an isomorphism $\mathbb{F}_{q^{d_i}} \cong \mathbb{F}_{q}(\alpha_i)$, so that the image of $a(u)$ in $\mathbb{F}_{q}(\alpha_i)$ is given by $a_i:=a(\alpha_i)$. Since all roots of $m(u)$ have multiplicity one, we have the following interpolation formula:
 $$
 a(u) = \sum\limits_{\beta \in \{\mathrm{roots\, of} \, m(u)\}}\, a(\beta) \dfrac{m(u)}{(u-\beta) m'(\beta)}.
 $$
Note that $m(u)/(u-\beta)$ is a monic polynomial of degree $d-1$ in $\bar{\mathbb{F}}_{q}[u]$. Thus, its coefficient of $u^{d-1}$ equals
 $$
 a_{d-1}=  \sum\limits_{\beta \in \{\mathrm{roots\, of} \, m(u)\}}\, \dfrac{a(\beta)}{m'(\beta)}.
 $$
Grouping this sum over the roots of $p_i(u)$, $i=1,\dots, r$, we obtain:
 $$
 a_{d-1}=\sum\limits_{i=1}^{r}\, \textrm{Tr}_{i}\, \Big( \dfrac{a(\alpha_i)}{m'(\alpha_i)} \Big) = \sum\limits_{i=1}^{r}\, \textrm{Tr}_{i}\, \Big( \dfrac{a_i}{m'(\alpha_i)} \Big) 
 $$
where $\textrm{Tr}_{i} :  \mathbb{F}_{q^{d_i}} \to  \mathbb{F}_{q}$ is the trace map. 
Next, 
 $$
 \mathrm{res}(m(u),b(u))  = \prod\limits_{\beta \in \{\mathrm{roots\, of} \, m(u)\}}\, b(\beta) = \prod\limits_{i=1}^{r}\, N_i(b(\alpha_i))
 $$
where $N:\mathbb{F}^{\times}_{q^{d_i}} \to  \mathbb{F}^{\times}_{q}$ is the norm map.
It follows from $a(u) b(u) =1 \pmod{m(u)}$ that $a(\alpha_i) b(\alpha_i)=1$. Thus, we obtain 
 $$
 \mathrm{res}(m(u),b(u))  =   \prod\limits_{i=1}^{r}\, N_i(a(\alpha_i))^{-1} =   \prod\limits_{i=1}^{r}\, N_i(a_i)^{-1}.
 $$
The summation over $a(u) \in R^{\times}$ corresponds to summation over $r$-tuples $(a_1,\dots,a_r) \in R^{\times}_1\times \dots \times R_{r}^{\times}$. Thus, we obtain:
 $$
S(m(u)) = \sum\limits_{(a_1,\dots,a_r) \in R^{\times}} \, \prod\limits_{i=1}^{r}  \chi\Big( N_i(a_i) \Big) \psi\Big(\textrm{Tr}_{i}\, \Big( \dfrac{a_i}{m'(\alpha_i)} \Big) \Big)
 $$
or 
 $$
S(m(u)) =   \prod\limits_{i=1}^{r}  \sum\limits_{a_i \in \mathbb{F}_{q^{d_i}}}\,  \chi\Big( N_i(a_i)\Big) \psi\Big(\textrm{Tr}_{i}\, \Big( \dfrac{a_i}{m'(\alpha_i)}\Big)\Big).  
 $$
Changing the summation variables $a_i \to a_i m'(\alpha_i)$, we obtain:
$$
S(m(u)) =   \Big( \prod\limits_{i=1}^{r} \chi(N(m'(\alpha_i))) \Big) \cdot \Big( \prod\limits_{i=1}^{r}   \sum\limits_{a_i \in \mathbb{F}_{q^{d_i}}}\,  \chi( N_i(a_i))  \psi(\textrm{Tr}_{i}\, ( a_i ))  \Big).
 $$
The first factor is
 $$
 \prod\limits_{i=1}^{r} \chi(N(m'(\alpha_i)))= \prod\limits_{\beta \in \{\mathrm{roots\, of} \, m(u)\}}\, \chi(m'(\beta)) = \chi(\mathrm{res}(m(u),m'(u))).
 $$
Moreover, $\mathrm{res}(m(u),m'(u)) = (-1)^\frac{d(d-1)}{2} D$, where $D$ is the discriminant of $m(u)$.

By the Hasse–Davenport lifting formula, we also obtain
 $$
 \sum\limits_{a_i \in \mathbb{F}_{q^{d_i}}}\,  \chi( N_i(a_i))  \psi(\textrm{Tr}_{i}\, ( a_i )) = (-1)^{d_i-1} G_q(\chi,\psi)^{d_i}. 
 $$
Altogether, this gives:
$$
S(m(u)) = (-1)^{d-r} \chi(-1)^{\frac{d(d-1)}{2}} G_q(\chi,\psi)^{d} \chi(D).
$$
Finally, we need to show that the sign can be written as
$
(-1)^{d-r} = \phi(D).
$
Note that $D=\delta^2$, where $\delta = \prod_{i<j} (\beta_i-\beta_j)$ and $\beta_i$ are the roots of $m(u)$. The $q$-Frobenius acts by a permutation on the set of roots, which we denote by $\sigma$. Thus,
$
\delta^q=\delta\, \textrm{sgn}(\sigma).
$
If $\textrm{sgn}(\sigma)=1$, then $\delta \in \mathbb{F}_q$ and $D$ is a full square. This means that 
$\textrm{sgn}(\sigma) = \phi(D)$. On the other hand, the $q$-Frobenius acts by a cyclic permutation of length $d_i$ on the set of roots of the irreducible factors $p_i(u)$. Thus,
$$
\textrm{sgn}(\sigma) = \prod_{i=1}^{r} (-1)^{d_i-1} =(-1)^{d-r},
$$
which finishes the proof in the square-free case. 

Next, assume that $m(u)$ is not square-free. Then we have the non-trivial nilpotent ideal
$$
J =\textrm{rad}(m(u))/(m(u)) \subset R.
$$
$J$ is an abelian group acting on $R^{\times}$ by shifts. Indeed, if $a\in R^{\times}$ and $j\in J$, then
$$
a + j = a(1+a^{-1} j),
$$
and since $a^{-1} j$ is nilpotent, $1+a^{-1} j$ is a unit, so $a + j \in R^{\times}$. Note that the multiplicative part of the superpotential is invariant under the $J$-action:
$$
N(b +j) = N(b) N(1+b^{-1} j) =N(b),
$$
where the last equality follows because $b^{-1} j$ is nilpotent. 

For any non-trivial $h \in R$, we can find $r\in R$ such that $\lambda(h r) \neq 0$. Indeed, if $h= c_s u^s +\dots \neq 0$, then
$$
\lambda(u^{d-1-s} h ) =c_{s} \neq 0.
$$
It follows that for any $j\in J$, we can find $r\in R$ such that $\lambda(j r) \neq 0$. Since $r j\in J$, the restriction $\left.\lambda\right|_{J}$ is a nontrivial additive character on $J$. It follows that the sum (\ref{sumw}) splits into sums over $J$-cosets, each of which gives
$$
\sum\limits_{j\in J}\, \chi(N(a+j)) \psi(\lambda(a+j))  = \chi(N(a)) \psi(\lambda(a)) \sum\limits_{j\in J}\,  \psi(\lambda(j))=0.
$$
The last sum vanishes because it is the sum of a non-trivial additive character over all elements of the abelian group $J$. Thus, if $m(u)$ is not square-free, then
$
S(m(u))=0,
$
which finishes the proof. 
\end{proof}

\subsection{Finite sums over fibers of $\pi$: deformed case}
Now, let us consider the deformed Coulomb branch ($f(u)\neq 1$) and the associated finite sum:
$$
S(m(u))=\sum_{t \in \pi^{-1}(m(u))}\, \Phi(t) =\sum\limits_{(a,b) \in \pi^{-1}(m(u))} \, \tilde{\chi}(b)^{-1}  \tilde{\psi}(a).
$$
By Proposition \ref{proprel}, the points in $\pi^{-1}(m(u))$ are given by two polynomials $a(u),b(u) \in \mathbb{F}_q[u]$ 
of degrees $\deg a(u)<d$, $\deg b(u)<d$, related by $a(u) b(u) = f(u) \pmod{m(u)}$. If we denote by $a$, $b$, and $f$ the images of $a(u)$, $b(u)$, and $f(u)$ in $R_{m(u)}$, the above sum takes the form
$$
S(m(u))=\sum\limits_{{(a,b) \in R^2_{m(u)},}\atop {a \cdot b =f}} \, \tilde{\chi}(b)^{-1}  \tilde{\psi}(a).
$$
Note that the relation $a b =f$ implies that $a$ and $b$ are invertible only if $f$ is invertible. As a result, the last sum cannot be immediately represented as an $R_{m(u)}$-Gauss sum. This, however, can be easily fixed by a small modification of the superpotential and an extension of the multiplicative characters.

Let $\chi'$ be a multiplicative character. Let us modify the superpotential (\ref{superpotdef}) as
$$
\Phi: t \mapsto   \chi'( \mathrm{res}(m(u),a(u))) \chi( \mathrm{res}(m(u),b(u)))^{-1} \psi( \lambda( a(u)) ). 
$$
Extend both characters by $\chi(0) =\chi'(0)=0$. With this modification, only points with invertible $a$ contribute to the above sum, and we obtain the following result:
\begin{thm} \label{thmfibdef}
$$
S(m(u))= (\chi\cdot\chi')(-1)^{\frac{d(d-1)}{2}} \phi(D)\, G_q(\chi\cdot\chi',\psi)^d (\chi\cdot \chi')(D) \chi^{-1}(\mathrm{res}( f(u),m(u) ))
$$
\end{thm}
\begin{proof}
Using the notation from (\ref{charlifts}), we write the sum as
$$
S(m(u))=\sum_{t \in \pi^{-1}(m(u))}\, \Phi(t) =\sum\limits_{{(a,b) \in R^2_{m(u)},}\atop {a \cdot b =f}} \, \tilde{\chi}'(a) \tilde{\chi}(b)^{-1}  \tilde{\psi}(a)
$$
with 
$
\tilde \chi' = \chi' \circ  \mathrm{res}(m(u),\cdot): R^{\times}_{m(u)} \to \mathbb{C}^{\times}.
$
Since $\chi'(0)=0$, only the points with $\mathrm{res}(m(u),a(u))\neq 0$ contribute. The condition $\mathrm{res}(m(u),a(u))\neq 0$ means that $a(u)$ and $m(u)$ are coprime, i.e., 
$$
a(u) r(u) + m(u) s(u) =1
$$
for some $r(u),s(u) \in \mathbb{F}_q[u]$, which means that $r(u)$ represents the inverse of $a$ in $R_{m(u)}$. Since $a$ is invertible, in the above sum we have $b=a^{-1} f$, and hence it takes the form
$$
S(m(u))=\tilde{\chi}^{-1}(f) \sum\limits_{a\in R^{\times}_{m(u)}} \, (\tilde{\chi}'\cdot\tilde{\chi})(a)  \tilde{\psi}(a).
$$
The result follows from Theorem \ref{ndth}. 
\end{proof}

\subsection{Selberg character sums}
Let $f(u)=u^{\alpha} (1-u)^{\beta}$ for natural numbers $\alpha,\beta$.
Let us fix a multiplicative character $\tau$ of order $q-1$. Since every other multiplicative character is a power of $\tau$, Theorem \ref{thmfibdef} implies, up to a sign,
$$
S(m(u))= G_q(\tau^c,\psi)^d \tau(D)^{\frac{q-1}{2}} \tau(D^c\, m(0)^a m(1)^b)
$$
for some $a,b,c \in \{0,\dots, q-1\}$.
The total sum over the Coulomb branch takes the form
$$
S = G_q(\tau^c,\psi)^d  \sum_{{\mathrm{monic}\, m(u) \in \mathbb{F}_q[u],}\atop {\deg m(u)=d}}\, \phi(D) \tau(D^c\, m(0)^a m(1)^b).
$$
We note that this is exactly the sum studied by Evans. 
\begin{thm}[Evans, \cite{Eva}]
If none of $a+b+(d-1+j)c$, $0 \leq j \leq d-1$, are divisible by $q-1$, then, up to a sign, we have

\begin{align}  \nonumber
&\sum_{{\mathrm{monic}\ m(u) \in \mathbb{F}_q[u],}\atop {\deg m(u)=d}}
\, \tau(D)^{\frac{q-1}{2}} \tau(D^c\, m(0)^a m(1)^b)
\\ \nonumber
&\qquad =
\prod\limits_{j=0}^{d-1}
\dfrac{
G_q(\tau^{a+jc},\psi)
G_q(\tau^{b+jc},\psi)
G_q(\tau^{c+jc},\psi)
\bar{G}_q(\tau^{a+b+(d-1+j)c},\psi)
}{
q\,G_q(\tau^c,\psi)
}.
\end{align}

where $\bar{G}_q$ denotes the complex conjugate of the Gauss sum. 
\end{thm}
Using this and the identity $\bar{G}_q(\chi,\psi) G_q(\chi,\psi) =q$, we arrive at
\begin{thm}
The superpotential sum over the Coulomb branch equals, up to a sign,
$$
S= \prod\limits_{j=0}^{d-1}\, \dfrac{G_q(\tau^{a+jc},\psi)G_q(\tau^{b+jc},\psi)G_q(\tau^{c+jc},\psi) }{{G}_q(\tau^{a+b+(d-1+j)c},\psi)}.
$$
\end{thm}

\section{Point count for $\mathcal{M}^{\times}_{f,d}(\mathbb{F}_{q}$) \label{pointcountse}} 
Let
 $|\mathcal{M}^{\times}_{f,d}(\mathbb{F}_{q})|$ be the number of $\mathbb{F}_{q}$-points of the Coulomb branch, and let 
$$
\mathcal{Z}_f(z) =  1+ \sum\limits_{d=1}^{\infty} |\mathcal{M}^{\times}_{f,d}(\mathbb{F}_{q})| z^d
$$
be the corresponding generating function. For a prime $p \in \mathbb{F}_q[x]$ we denote by $\nu_{p}(f)$ the corresponding valuation of $f$ (i.e., the power of $p$ appearing in the prime factorization of $f$).
\begin{thm}
 \bean \label{pcthm}
 \mathcal{Z}_f(z)= \dfrac{(1-z q)^2}{(1-z) (1-z q^2)}\, \prod\limits_{p\mid f, \atop {p\neq x}} \dfrac{1-(q z )^{(\nu_{p}(f)+1)\deg(p)}}{1-(q z)^{\deg(p)}}.
\eean
where the product runs over primes $p\in \mathbb{F}_q[x]$, $p\neq x$, which divide $f$. 

\end{thm}

\begin{proof}
Let $N_f(m)$ be the number of $\mathbb{F}_q$-points in the fiber $\pi^{-1}(m)$, so that
 $$
 |\mathcal{M}^{\times}_{f,d}(\mathbb{F}_{q})| = \sum_{\deg (m) =d }\, N_f(m),
 $$
 and 
$$
  \mathcal{Z}_f(z)=  \sum_{m}\, N_f(m) z^{\deg(m)},
$$
where the sum is over all monic $m$ with non-zero constant term. By definition,
$$
N_f(m) = |\{(a,b) \in R^2_{m}: a b =\bar{f} \}|,
$$
where $\bar{f}$ denotes the image of $f$ in $R_{m}=\mathbb{F}_q[u]/(m)$. By the Chinese remainder theorem, if $m$ has prime factorization 
$m = p_1^{e_1} \cdots p_l^{e_l}$, then
$
R_{m} = R_{p_1^{e_1}} \times \cdots  \times R_{p_l^{e_l}}
$
and $N_f(m) = N_f(p_1^{e_1}) \cdots N_f(p_l^{e_l})$. Thus,
$$
  \mathcal{Z}_f(z)= \prod_{{p = \mathrm{prime}} \atop {p\neq x}} L_{f,p}(z),
$$
where the local factor is
$$
L_{f,p}(z) = 1+\sum\limits_{e=1}^{\infty}\, N_f(p^e) z^{e \deg(p)}.
$$
The condition $p\neq x$ above is the same as the non-vanishing of the constant term. 

From (\ref{fibsize}) below, for $f=1$ we obtain $N_{1}(p^e) = (q^{\deg(p)}-1) q^{(e-1) \deg(p)}$ and therefore 
$$
L_{1,p}(z)= \dfrac{1-z^{\deg(p)}}{1-q^{\deg(p)} z^{\deg(p)} }.
$$
By the Euler product formula,
$$
\prod\limits_{{p = \textrm{prime}}}\, \dfrac{1}{1-z^{\deg(p)}} = \dfrac{1}{1- q z}.
$$
Thus,
$$
\mathcal{Z}_1(z) =  \prod\limits_{{p = \textrm{prime},}\atop {p\neq x} } \, \dfrac{1-z^{\deg{p}}}{1-q^{\deg{p}} z^{\deg{p}} } = \dfrac{1-z q}{1-z}  \prod\limits_{{p = \textrm{prime}}} \, \dfrac{1-z^{\deg{p}}}{1-q^{\deg{p}} z^{\deg{p}} } =\dfrac{(1-z q)^2}{(1-z) (1-z q^2)}.
$$
Assume $\nu_{p}(f)=k$. Then, from (\ref{fibsize}), we obtain
$$
L_{f,p} = 1+ \sum\limits_{e=1}^{k} \Big(q^{e d} + e(q^{d}-1) q^{d(e-1)}\Big) z^{d e} +\sum\limits_{e=k+1}^{\infty}\, (k+1) (q^{d}-1) q^{(e-1)d} z^{d e},   
$$
where $d=\deg p$. A direct calculation gives
$$
\dfrac{L_{f,p}(z)}{L_{1,p}(z)} = 1+ z^{d} q^{d} +\dots+ z^{k d} q^{k d} = \dfrac{1-z^{d (k+1)} q^{d(k+1)} }{1-z^{d} q^{d}}.
$$
Thus,
$$
\mathcal{Z}_f(z)= \prod\limits_{p=\textrm{prime} \atop {p\neq x}} \, L_{f,p}(z) = \Big(\prod\limits_{p=\textrm{prime} \atop {p\neq x}} \, L_{1,p}(z) \Big)\Big(  \prod\limits_{p=\textrm{prime} \atop {p\neq x}} \dfrac{1-(q^{\deg(p)} z^{\deg(p)} )^{\nu_p(f)+1}}{1-q^{\deg(p)} z^{\deg(p)}}\Big). 
$$
The first product here is $\mathcal{Z}_{1}(z)$ computed above, and all factors in the second product are equal to $1$ except those for which $p\mid f$. This finishes the calculation. 

\end{proof}

\begin{lem}
Let $p\in \mathbb{F}_q[x]$ be a prime and let $R_{p^e} = \mathbb{F}_q[x]/(p^e)$. Then
$$
 |\{a\in R_{p^e}: \nu_p(a)=r\}| = (q^d-1) q^{d(e-r-1)} 
$$
where $d=\deg(p)$.
\end{lem}
\begin{proof}
Consider the ideal $p^r R$. Note that $p^r R \cong \mathbb{F}_q[x]/(p^{e-r})$, from which we obtain
\bean \label{idealsize}
|p^r R_{p^e}| =q^{d(e-r)}.
\eean 
Thus,
$$
 |\{a\in R_{p^e}: \nu_p(a)=r\}| =  |\{a\in R_{p^e}: \nu_p(a)\geq r\}| - |\{a\in R_{p^e}: \nu_p(a)\geq r+1\}|
 = |p^r R_{p^e}|-|p^{r+1} R_{p^e}|
$$
\end{proof}

\begin{thm} We have
\bean \label{fibsize}
N_f(p^e)  = \left\{\begin{array}{ll}
(s+1) (Q-1) Q^{e-1}, & s<e,\\
Q^e + e(Q-1) Q^{e-1}, & s=e
\end{array}\right.
\eean 
where $s=\min(\nu_p(f),e)$ and $Q=q^{\deg(p)}$.    
\end{thm}
\begin{proof}
First, assume $\bar{f} = p^s u_f$ with $s<e$ and $u_f \in R^{\times}$. Let us fix $a= p^r u_a \in R$ for $u_a \in R^{\times}$. We count solutions $b$ satisfying $a b =\bar{f}$. From this equation, we see that solutions exist only if $r\leq s$, i.e., $r=0,1,\dots, s$. The solutions $b$ of $a b =\bar{f}$ are of the form
 $$
 b=b_0 + \ker \mu_a
 $$
where $b_0 = p^{s-r} u_f/u_a$ is a particular solution and $\mu_a$ is the operator of multiplication by $a$ in $R$. 
Thus, the number of solutions is $|\ker \mu_a|$. We note that 
 $$
 \ker \mu_a = p^{e-r} R
 $$
thus, from (\ref{idealsize}), we see that $|\ker \mu_a|=Q^r$. From the lemma above, the number of $a$'s with valuation $r$ is $(Q-1) Q^{e-r-1}$. Thus, the total number of pairs is $(Q-1) Q^{e-1}$. Summing over all possible $r=0,1,.\dots s$ gives
 $$
 (s+1)(Q-1) Q^{e-1}
 $$
as needed. 

Next, if $s=e$ we have $\bar{f}=0$ and we count solutions $a b=0$. If $a\neq 0$, then, as above, $a=p^r u_s$ for $r=0,1,\dots, e-1$. Arguing exactly as above, we obtain 
$
e(Q-1) Q^{e-1}
$
solutions. If $a=0$, then any $b$ gives a solution, which gives $Q^e$ more solutions. The total number of solutions is
$$
Q^e+e(Q-1) Q^{e-1}
$$
as needed. 
\end{proof}

\end{document}